\RequirePackage{fix-cm}
\documentclass[pdflatex,sn-mathphys-num]{sn-jnl}
\usepackage{etoolbox}

\usepackage{siunitx}
\usepackage{graphicx}
\usepackage{lipsum}
\usepackage{multirow}

\RequirePackage{fix-cm}
\usepackage{graphicx}
\usepackage{multirow}
\usepackage{amsmath,amssymb,amsfonts}
\usepackage{amsthm}
\usepackage{mathrsfs}
\usepackage[title]{appendix}
\usepackage{xcolor}
\usepackage{textcomp}
\usepackage{manyfoot}
\usepackage{booktabs}
\usepackage{algorithm}
\usepackage{algorithmicx}
\usepackage{algpseudocode}
\usepackage{listings}

\usepackage{braket}
\usepackage{amsfonts}
\usepackage{bm}

\theoremstyle{thmstyleone}
\newtheorem{lemma}{Lemma}
\newtheorem{proposition}{Proposition}

\newtheorem{corollary}{Corollary}
\newtheorem*{example*}{Example}
\newtheorem*{remark*}{Remark}
\newtheorem*{remarks*}{Remarks}

\newcommand{\dd}{\,\mathrm{d}}

\usepackage{float}
\newcommand{\nn}{\nonumber}

\newcommand{\NN}{{\mathbb N}}

\newcommand{\modu}[1]{\left|{#1}\right|}
\newcommand{\B}[1]{\left({#1}\right)}
\newtheorem{dfn}{Definition}

\begin{document}

\title[Relative entropy]{\centering {Higher-order spectral form factors of \\ circular unitary ensemble}}
\author{\fnm{Sohail\textsuperscript{\dag}},~\fnm{Youyi} \sur{Huang}\textsuperscript{*},~\fnm{Lu} \sur{Wei}\textsuperscript{\dag}}
\affil{\orgname{\center \textsuperscript{\dag}Texas Tech University} \\ \orgaddress{\city{Lubbock}, \state{Texas} \postcode{79409}, \country{USA}}}
\affil{\orgname{\center \textsuperscript{*}University of Central Missouri} \\ \orgaddress{\city{Warrensburg}, \state{Missouri} \postcode{64093}, \country{USA}}}
\maketitle
\begin{abstract}

Spectral form factor (SFF), one of the key quantity from random matrix theory, serves as an important tool to probe universality in disordered quantum systems and quantum chaos. In this work, we present exact closed-form expressions for the second- and third-order SFFs in the circular unitary ensemble (CUE), valid for all real values of the time parameter, and analyze their asymptotic behavior in different regimes. In particular, for the second-order SFF, we derive an exact closed-form expression in terms of polygamma functions. In the limit of infinite matrix size, and when the time parameter is restricted to integer values, the second-order SFF reproduces the standard result established in earlier studies. When the time parameter is of order one relative to the matrix size, we demonstrate that the second-order SFF grows logarithmically with the ensemble dimension. For the third-order SFFs, a closed-form result in a special case is obtained by exploiting the translational invariance of CUE. 
\end{abstract}
\newpage

\section{Introduction and Main Results}\label{sec:intro}
Quantum chaos~\cite{gutzwiller1971periodic,heller1984scars,aurich1988hadamard,berry1987quantumchaology,Haake,stockmann1999quantumchaos,eckhardt1992semiclassical} provides a way to describe a chaotic classical system within the laws of quantum physics. In classical physics, the chaotic behavior of a physical system can be identified by looking at its phase space: for chaotic dynamics, the distance between any two trajectories in the phase space grows exponentially, whereas in the case of regular motion the growth in the distance may follow the power law, but never becomes exponential~\cite{strogatz2018nonlinear,lichtenberg1992regular}. Since the phase space looses its meaning for quantum systems because of uncertainly principle~\cite{Heisenberg1927,sakurai2017,robertson1929}, the characterization of quantum chaos in terms of the distance between trajectories is not applicable. The appropriate quantities to probe the chaotic behavior of quantum systems are based on the spectrum, eigenvectors of the Hamiltonian associated with the system and temporal evolution of the expectation values of some suitable observables. As one of the primary focuses of random matrix theory is the study of spectral properties of random matrices, it provides a powerful framework to study the quantum signature of chaos. SFF being intrinsically defined in terms of the spectrum serves as one of the useful quantities to probe the chaotic nature of a quantum system~\cite{Berry,Haake,Nick}. 

Although, in random matrix theory, large system size limits are considered very often because it greatly simplifies the mathematical analysis and reveals universal statistical properties that are independent of specific system details, it is also crucial to have finite dimensional results to model real-world systems with limited dimensions, exhibiting system specific characteristics and critical deviations from the universal behavior. In this work, we will be interested in studying the exact SFF and its higher order generalizations in CUE of finite size.
 
\subsection{Problem formulation}
CUE is a unitary invariant ensemble, i.e., it is the ensemble of $N \times N$ unitary matrices $U$ with the probability $P(U)dU$ of a randomly chosen unitary matrix $U$ being in the volume element $dU$ is invariant under the transformation $U \mapsto VUW$, where $V$ and $W$ are $N \times N$ unitary matrices~\cite{mehta1991random}. The eigenvalues of a matrix belonging to CUE, lie on the unit circle of the complex plane, i.e., all the eigenvalues are of the form $\lambda_j=e^{i \theta_j}$, where $\theta_j$ is the eigenangle associated with the eigenvalue $\lambda_j$. The joint probability distribution of the eigenangles in CUE is given by~\cite{mehta1991random,Haake}
\begin{align}
        p(\theta_1, \theta_2,...,\theta_N)=\frac{1}{\mathcal{Z}} \prod_{1 \leq j <k \leq N}\big|e^{i \theta_j} -e^{i \theta_k} \big|^2 \label{eq:JPDF of CUE}
\end{align}
    with $\mathcal{Z}$ being the normalization constant. Notice that the above distribution is invariant under the permutation of the indices $\{1,2,...,N\}$. The $n$-point correlation function $\rho(\theta_1,\theta_2,...,\theta_n)$ defined by~\cite{mehta1991random}
\begin{align}
        \rho(\theta_1,\theta_2,...,\theta_n)&:=\frac{N!}{(N-n)!}\prod_{j=n+1}^N \frac{1}{2 \pi}  \int_{0}^{2 \pi}  p(\theta_1, \theta_2,...,\theta_N) \dd \theta_j,
\end{align}
    can be written in terms of the kernel $K_N(\theta_j, \theta_k)$ as~\cite{mehta1991random} 
    \begin{align}
        \rho(\theta_1,\theta_2,...,\theta_n)&=\det [K_N(\theta_j, \theta_k)]_{j,k=1}^n,
    \end{align}
    where
    \begin{align}
    K_N(\theta,\phi)= \frac{1}{2 \pi}\sum_{l=0}^{N-1} e^{i l (\theta-\phi)}. \label{CUE_kernel}
\end{align}
The CUE kernel given by the above equation satisfies the following reproducing property~\cite{mehta1991random}
\begin{align}
    \int_{0}^{2\pi} K_N(\theta, \gamma)  K_N(\gamma, \phi) d\gamma=K_N(\theta, \phi).
\end{align}
 
One of the key distinguishing features of CUE is the translational invariance of the joint probability distribution function as given by Eq.~(\ref{eq:JPDF of CUE})~\cite{mehta1991random}, i.e.,
\begin{align}
    p(\theta_1+\phi, \theta_2+\phi,...,\theta_N+\phi)=p(\theta_1, \theta_2,...,\theta_N).
\end{align}
The kernel $K_N(\theta, \phi)$ being a function of $(\theta-\phi)$, the translational invariance of CUE is explicit.

The second-order SFF, one of our main objects of study, is defined in Definition~\ref{dfn:Second-order SFF} in terms of the microscopic eigenvalue density defined as follows. 
\begin{dfn}[\cite{Forrester2021stat}]\label{dfn:microscopic eigenvalue density}
    For a given ensemble of $N \times N$ random matrices, the microscopic eigenvalue density is defined by $n(\lambda):=\sum_{i=1}^N \delta(\lambda-\lambda_i)$ with the set of eigenvalues being denoted by $\{ \lambda_i \}_{i=1}^N$. The density-density correlation function $N_{(2)}(\lambda, \lambda')$ is defined as the covariance between $n(\lambda)$ and $n(\lambda')$, i.e.,
    \begin{align}
        N_{(2)}(\lambda, \lambda'):= \langle n(\lambda) n(\lambda') \rangle -\langle n(\lambda) \rangle \langle n(\lambda') \rangle,
    \end{align}
     where $\langle \cdot \rangle$ denotes the ensemble average.
\end{dfn}
\begin{dfn}[\cite{Forrester2021stat}] \label{dfn:Second-order SFF}
    The second-order SFF $S_N(k)$ is defined as the Fourier transform of $N_{(2)}(\lambda, \lambda')$ at the point $(k,-k)$, i.e.,
    \begin{align}
        S_N(k):=\int_{0}^{2 \pi} \int_{0}^{2 \pi}e^{ik(
        \lambda-\lambda')} N_{(2)}(\lambda, \lambda') d \lambda d\lambda' .\label{eq:Second order SFF CUE}
    \end{align}
\end{dfn}

The second-order SFF as defined above can be expressed as the covariance between the linear statistics $\sum_{i=1}^N e^{ik\lambda_{i}}$ and its complex conjugate $\sum_{i=1}^N e^{-ik\lambda_{i}}$~\cite{Forrester2021stat}. When $k \in \mathbb{N}$, the second-order SFF in CUE is well-known in the literature and is given by the following formula~\cite{Haake,HSW}
\begin{align}
    S_N(k)=\mathrm{min}(|k|,N). \label{SFF_int}
\end{align}
The Eq.~(\ref{SFF_int}) shows a ramp behavior for the interval $0 \leq k \leq N$ and a plateau behavior for $k >N$. 

In this work, we derive the exact closed-form formula for the second-order SFF in CUE for all real times in terms of the polygamma functions and analyze their asymptotic behavior in different regimes. Then, utilizing the translational invariance in CUE, we define third order SFF as the two-variable Fourier transform of the triple-density correlation function and derive closed-form formulas for the third-order SFF when one of the variables is assuming integer values. 

The paper is organized as follows: In Section~\ref{sec:Main results}, we discuss our main results. This Section is divided in two parts. In Section~\ref{subsec:Second order SFF}, we provide the closed-form formula for the second order SFF at any arbitrary time, with the asymptotic expansions at various limits. In Section~\ref{subsec:Higher order SFF}, we provide a semi-closed form expression for the third order SFF and a closed form expression in a special case. In Section~\ref{sec:Proofs of the main results}, we prove our main results. In Section~\ref{Discussions}, we discuss our findings and possible future directions emerging from our work.
\section{Main Results} \label{sec:Main results}
In this section, we introduce our main results. Proposition~\ref{thm:SFF} provides the exact closed-form formula for the second-order SFF which is valid for any size of the ensemble and any argument. Corollary~\ref{prop:k=Nt asymptotic}, Corollary~\ref{prop:k_O(1)} and Corollary~\ref{prop:Arbitrary_point_x_0} deal with the asymptotics of the second-order SFF in various limits. In proposition~\ref{prop: 3rd_order_SFF_at_(k,-k)} and Proposition~\ref{3rd SFF k,floor k}, we provide a semi-closed form expression for the third-order SFF and a closed-form expression in a special case, respectively.

Since all our results are expressed in terms of polygamma functions, we begin with a brief discussion of their basic properties. The $k$-th order polygamma function is defined as 
\begin{align}
    \psi_{k}(x):=\frac{\dd^{k+1}}{\dd x^{k+1}} \ln \Gamma(x),
\end{align}
where $\Gamma(x)$ is the gamma function. $\psi_{0}(x)$ and $\psi_{1}(x)$ are called the digamma and trigamma function, respectively. Polygamma function satisfy the difference equation
\begin{align}
    \psi_{k}(x+N)-\psi_{k}(x)=(-1)^k k!\sum_{i=0}^{N-1} \frac{1}{(x+i)^{k+1}}.\label{eq:Telescopic polygamma}
\end{align}
For $x >0$ and $x \notin \mathbb{N}$, the reflection formula~\cite{Luke} 
\begin{align}\label{eq:reflection}
   \psi_0(-x)=\psi_0(x+1)+\pi \cot(\pi x)
\end{align}
relates $\psi_{0}(x)$ and $\psi_{0}(-x)$.
\subsection{Second-order SFF} \label{subsec:Second order SFF}
\begin{proposition}\label{thm:SFF}
    The closed-form expression for the second-order SFF in CUE, as defined in Eq.~(\ref{eq:Second order SFF CUE}), is given by  
    \begin{align}
    S_N(k)&= \mathrm{min}\left(N,|k|- \frac{\sin(2\pi |k|)}{2\pi}\right)+\frac{\sin^2(\pi |k|)}{\pi^2} \bigg( \psi_0(N+|k|) + (N+|k|) \psi_{1}(N+|k|) \nn\\
    &\quad+ \psi_0\left(\left|N-|k|\right|+1\right) +\left|N-|k|\right| \psi_{1}\left(\left|N-|k|\right|+1\right)
    -2|k| \psi_{1}(|k|+1)\nn\\
    &\quad- 2\psi_0(|k|+1)\bigg).   \label{exact_SFF_1}
 \end{align}
 \end{proposition}
The proof can be found in Section~\ref{sec:Proofs of the main results}. For discrete times $k \in \NN$, the above expression simplifies to 
\begin{equation}
     S_N(k) =\mathrm{min}\left(N,|k|\right) ,\label{exact_SFF_2}
\end{equation}
thereby retrieving the well-known result~(\ref{SFF_int}).  This integer-time result captures the broad ramp–plateau behavior as seen from the orange curve in Figure~\ref{fig:2ndSFF},  but provides only a coarse, discretized view of the dynamics. In contrast, the blue curve shown here represents our newly derived complete result~(\ref{exact_SFF_1}), valid for all times. The exact expression reveals the full continuous structure of the second-order SFF, capturing the intermediate behavior between integer points and providing a more detailed characterization of the ramp and the plateau. By going beyond the integer-time approximation, our result provides more accurate understanding of the temporal evolution of spectral correlations in the CUE.
\begin{figure}[H]
    \centering
    \includegraphics[width=0.6\linewidth]{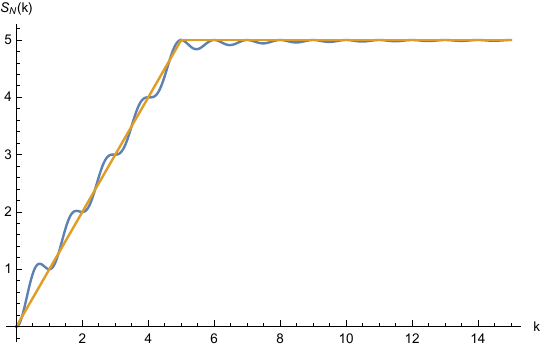}
    \caption{Second-order SFF with $N=5$. The Blue curve represents the closed-form expression of the SFF given in Eq.~(\ref{exact_SFF_1}), while the orange curve shows the discrete SFF at integer $k$ as given by Eq.~(\ref{exact_SFF_2}).}
    \label{fig:2ndSFF}
\end{figure}

The next two corollaries describe the asymptotic behavior of the second-order SFF given in Eq.~(\ref{exact_SFF_1}) in the regimes $k=\mathcal{O}(N)$ and $k=\mathcal{O}(1)$. Proofs are provided in Section~\ref{sec:Proofs of the main results}.
\begin{corollary} \label{prop:k=Nt asymptotic}
    Let $k$ be of the order of $N$, i.e., $k=t N$ with $t$ being the local coordinate which does not scale with $N$. Then, in the limit $N \to \infty$, we have 
     \begin{align}
    S_N(t N)= \min\left(N,N|t|- \frac{\sin(2\pi N|t|)}{2\pi}\right)+ \frac{\sin^2(\pi N|t|)}{\pi^2} \ln \frac{\big|1-|t|^2\big|}{|t|^2}+\mathcal{O}\left(N^{-2}\right).
 \end{align}
\end{corollary}
As a direct consequence of the above corollary, we obtain
\begin{align}
    \lim_{N \to \infty} \frac{1}{N} S_N(t N)= \min\left(1,|t|\right),
\end{align}
 reproducing the well-established large-$N$ behavior~\cite{Forrester2021stat}.
\begin{corollary} \label{prop:k_O(1)}
    Let $k$ be of the order $1$ with respect to $N$. Then we have
    \begin{align}
        S_N(k) &=\frac{2\sin^2(\pi |k|)}{\pi^2}\ln N+h(k)+\mathcal{O}\left(N^{-2}\right),
    \end{align}
    where 
    \begin{align}
        h(k)&:=\left(1-\frac{2\sin^2(\pi |k|)}{\pi^2} \psi_{1}(|k|+1)\right)|k|-\frac{2\sin^2(\pi |k|)}{\pi^2}\psi_0(|k|+1)\nn\\
        &\quad- \frac{\sin(2\pi |k|)}{2\pi}+ \frac{2\sin^2(\pi |k|)}{\pi^2}.
    \end{align}
\end{corollary}  
For $k$ of the order one, the asymptotic form of the second-order SFF reveals three distinct contributions: a leading term exhibiting logarithmic growth with the ensemble dimension $N$, a finite term $h(k)$ of order one, and a subleading correction that scales as $N^{-2}$.

To closely examine the behavior of the second-order SFF around an arbitrary point $x_0:=\frac{k}{N}$, in the rescaled coordinate, we define the following object of interest: 
\begin{align}
    \Delta(x_0,\tau):= S_N\left(N x_0+\tau\right)-\min(N,Nx_0), \label{eq:Fluctuation} 
\end{align}
where $S_N(k)$ is the second-order SFF as given in Eq.~(\ref{exact_SFF_1}), and $\tau <N$ is a small parameter that measures the distance from the point $x_0$. The condition $\tau < N$ ensures that this shift is small relative to the system size. The quantity $\Delta(x_0,\tau)$ essentially provides information about the local fluctuations of the second-order SFF around a given point $x_0$.
\begin{corollary}\label{prop:Arbitrary_point_x_0}
Case-I: Define
\begin{align}
    \mu_n (x_0):=\frac{1}{(1+x_0)^n}+\frac{(-1)^n}{(1-x_0)^n}-\frac{2}{x^n_0}.
\end{align}
Then for $0<x_0<1$,
    \begin{align}
      \Delta(x_0,\tau)&=\tau- \frac{\sin(2\pi (N x_0+\tau))}{2\pi}+\frac{\sin^2(\pi (N x_0+\tau))}{\pi^2 } \bigg( \ln\frac{1-x_0^2}{x_0^2}+ \frac{\tau}{N}\mu_1(x_0)\nn\\
    &\quad+\frac{1-6\tau^2}{12N^2} \mu_2(x_0)+\frac{\tau(2\tau^2-1)}{6N^3}\mu_3(x_0)\bigg)+\mathcal{O}\left(N^{-4}\right). 
    \end{align}
    
   Case-II: When $x_0=0$, we have 
    \begin{align}
      \Delta(0,\tau)&=|\tau|- \frac{\sin(2\pi |\tau|)}{2\pi}+\frac{\sin^2(\pi |\tau|)}{\pi^2 } \bigg(2\ln N+\frac{1-6\tau^2}{6N^2}-2\psi_0(\modu{\tau})\nn\\
      &\quad-2 \modu{\tau} \psi_{1}(\modu{\tau})+2\bigg)+\mathcal{O}\left(N^{-4}\right).
    \end{align}
    Case-III: When $x_0=1$, we have 
    \begin{align}
      \Delta(1,\tau)&= \tau- \frac{\sin(2\pi \tau)}{2\pi}+\frac{\sin^2(\pi \tau)}{\pi^2 } \bigg(-\ln N- \frac{3\tau}{2N}-\frac{7(1-6\tau^2)}{48N^2}-\frac{5(2\tau^3-\tau)}{16N^3}\nn\\
    &\quad+\psi_0(\modu{\tau})+\modu{\tau} \psi_{1}(\modu{\tau})-1+ \ln 2\bigg)+\mathcal{O}\left(N^{-4}\right)
    \end{align} 
    for $\tau <0$,
    and
    \begin{align}
       \Delta(1,\tau) &=\frac{\sin^2(\pi \tau)}{\pi^2 } \bigg(-\ln N- \frac{3\tau}{2N}-\frac{7(1-6\tau^2)}{48N^2}-\frac{5(2\tau^3-\tau)}{16N^3}\nn\\
    &\quad+\psi_0(\modu{\tau})+\modu{\tau} \psi_{1}(\modu{\tau})-1+ \ln 2\bigg)+\mathcal{O}\left(N^{-4}\right)
    \end{align}
      for $\tau >0$.
\end{corollary}
\subsection{Higher-order SFF} \label{subsec:Higher order SFF}
While the second-order SFF is an invaluable tool to probe the chaotic nature of a quantum system, it relies completely on the two-point correlation function and therefore captures only short-range features such as level repulsion and local spectral fluctuations. To obtain a detailed and more comprehensive picture of quantum chaos, it is desirable to go beyond two-point correlations and incorporate higher-order correlation functions in order to probe long-range structures such as spectral rigidity, suppressed fluctuations, and universal scaling laws. By extending the analysis to three-point and higher correlations, one gains access to information that remains invisible at the level of pairwise correlations, enabling a deeper characterization of chaotic dynamics and a clearer distinction among chaotic, integrable, and pseudorandom systems. Here, we aim to find an exact closed-form expression for the third-order SFF in CUE.

One can generalize the density-density correlation function to third-order as the cumulant between $n(\lambda),n(\lambda')$ and $n(\lambda'')$, i.e.,
\begin{align}
    N_{(3)}(\lambda, \lambda', \lambda''):=\kappa(n(\lambda),n(\lambda'),n(\lambda'')),
\end{align}
where $\kappa$ denotes the joint cumulant and $n(\lambda)$ is the microscopic eigenvalue density (see Definition~\ref{dfn:microscopic eigenvalue density}).
It follows that $N_{(3)}$ is translationally invariant for CUE, i.e., $ N_{(3)}(\lambda, \lambda', \lambda'')= N_{(3)}(\lambda-\lambda'', \lambda'-\lambda'', 0)$. The translational invariance greatly simplifies the mathematical form of the third-order SFF, and in certain cases it enables to provide closed-form expressions by decoupling the multiple summations arising from the summation form of the CUE kernel (see Eq.~(\ref{CUE_kernel})). In CUE, the explicit form of triple-density correlation function $N_{(3)}(\lambda,\lambda',0)$ is given by
    \begin{align}
      N_{(3)}(\lambda,\lambda',0)&= \rho(\lambda, \lambda',0)+\rho(\lambda',0) \delta(\lambda- \lambda')+\rho(-\lambda',0) \delta(\lambda)+\rho(\lambda, \lambda') \delta(\lambda')\nn\\
      &+\frac{N}{2\pi} \delta(\lambda) \delta(\lambda')+ 2 \left(\frac{N}{2\pi}\right)^3-\frac{N}{2\pi} \rho(\lambda, \lambda')-\left(\frac{N}{2\pi}\right)^2 \delta(\lambda -\lambda')\nn\\
      &-\frac{N}{2\pi} \rho(\lambda, 0)-\left(\frac{N}{2\pi}\right)^2 \delta(\lambda )-\frac{N}{2\pi} \rho(\lambda', 0)-\left(\frac{N}{2\pi}\right)^2 \delta(\lambda').
    \end{align}
    
    Now, we define the third-order SFF in terms of the triple-density correlation function $N_{(3)}(\lambda,\lambda',0)$ as follows.
\begin{dfn}
     The third-order SFF is defined as the Fourier transform of $N_{(3)}(\lambda, \lambda',0)$ at $(k,-k)$, i.e.,
    \begin{align}
        S_{N}^{(3)}(k,-k):=\int_{0}^{2 \pi} \int_{0}^{2 \pi}e^{ik(
        \lambda-\lambda')} N_{(3)}(\lambda, \lambda',0) \dd \lambda \dd \lambda' .\label{eq:Third-Order_SFF}
    \end{align}
\end{dfn}
The third-order SFF, as defined in Eq.~(\ref{eq:Third-Order_SFF}), is in general a complex-valued function. In the following proposition, we present a semi-closed form expression for the real part.
\begin{proposition} \label{prop: 3rd_order_SFF_at_(k,-k)}
The real part of third-order SFF $S_{N}^{(3)}(k,-k)$ is given by
  \begin{align}
      \mathrm{Re}\left(S_{N}^{(3)}(k,-k)\right)&= f(N,k)+f(N,-k), \label{eq:Semi-closed third-order SFF}
  \end{align}
  where 
\begin{align}
    f(N,k)&:=-\frac{\sin ^2(\pi  k) }{\pi ^3}\sum _{l=1}^N \psi_0 (k+l) \psi_0 (-k-l+N+1)\nn\\
    &\quad+\frac{\sin (\pi  k)}{\pi ^3 k}\bigg(\left(-2 k^2-3 k N+4 \pi ^2 N (k+N)+k-N^2\right) \sin (\pi  k)\nn\\
    &\quad-2 \pi  \left(2 \pi ^2-1\right) k (k+N) \cos (\pi  k)\bigg)\psi_0 (k+N)+\frac{k+N}{\pi ^3}\sin ^2(\pi  k)\nn\\
    &\quad \times \psi_0 (k+N)^2+\frac{4 \pi ^2-1}{\pi ^3}\sin (\pi  k) (\pi  k \cos (\pi  k)-N \sin (\pi  k))\psi_0 (k)\nn\\
    &\quad+\frac{\left(4-\pi ^2+4 \pi ^4\right) N-2}{2 \pi ^3}\sin ^2(\pi  k)+\frac{3 \left(1-4 \pi ^2\right) N}{4 \pi}\nn\\
    &\quad-\frac{2 k^2+N^2-4 \pi ^2 N^2}{2 \pi ^2 k}\sin (\pi  k) \cos (\pi  k).
\end{align}
\end{proposition}
Note that the difference between the real parts of the third-order SFF at $(k,-k)$ and $(k+1,-k-1)$, denoted by 
\begin{align}
    \Delta \mathrm{Re}\left(S^{(3)}_{N}(k)\right):=\mathrm{Re}\left(S_{N}^{(3)}(k+1,-k-1)-S_{N}^{(3)}(k,-k)\right),
\end{align}
admits a closed-form expression given by
    \begin{align}
        \Delta \mathrm{Re}\left(S^{(3)}_{N}(k)\right)&= \sin^2 (\pi  k) (\psi_0 (k-N+1)-\psi_0 (k+N+1)) \nn\\&\quad \times \left(\frac{\left(4 \pi ^2-1\right) N^2 }{\pi ^3 k (k+1)}+\frac{2}{\pi ^3}\psi_0 (k+1)\right)\nn\\
        &\quad+\frac{\sin ^2(\pi  k)}{\pi ^3}\left(\psi_0 (k+N+1)^2-\psi_0 (k-N+1)^2\right)\nn\\
        &\quad+2 \sin (2 \pi  k) (-\psi_0 (k-N+1)-\psi_0 (k+N+1)\nn\\
        &\quad+2 \psi_0 (k+1)).
    \end{align}

The third-order SFF $S_{N}^{(3)}(k,-k)$ presented in Eq.~(\ref{eq:Semi-closed third-order SFF}) is not a closed-form expression. When $k\in \mathbb{N}$, the expression simplifies to a closed form, however this a highly restrictive scenario. A more flexible and natural quantity to consider is $S_{N}^{(3)}(k,-\lfloor{k} \rfloor)$, where $\lfloor{k} \rfloor$ denotes the floor of $k$. The following proposition yields a closed-form expression for $S_{N}^{(3)}(k,b)$ defined by 
\begin{align}
        S_{N}^{(3)}(k,b):=\int_{0}^{2 \pi} \int_{0}^{2 \pi}e^{i(
        k\lambda+b\lambda')} N_{(3)}(\lambda, \lambda',0) \dd \lambda \dd \lambda' \label{eq:Third-Order_SFF (k,b)},
    \end{align}
    where $b \in \mathbb{N}$.
\begin{proposition} \label{3rd SFF k,floor k}
    The real part of the third-order SFF $S_{N}^{(3)}(k,b)$, as defined in Eq.~(\ref{eq:Third-Order_SFF (k,b)}), is given by 
   \begin{align}
       \mathrm{Re}\left(S_{N}^{(3)}(k,b)\right)&= \frac{N-\max (0,N-b )}{2 \pi }+\frac{\sin (2 \pi  k)}{4 \pi ^2}\bigg(\Theta (-b+N-1) \nn\\
       &\quad \times \bigg(2 (b+k-N) \psi_0 (b+k-N+1)-2 (b+k) \psi_0 (b+k+1)-2 k \psi_0 (k+1)\nn\\
       &\quad +2 (k+N) \psi_0 (k+N+1)\bigg)-(b+k-N) \psi_0 (b+k-N+1)\nn\\
       &\quad-(b+k+N) \psi_0 (b+k+N+1)-(k+N) \psi_0 (k+N+1)\nn\\
       &\quad+2 (b+k) \psi_0 (b+k+1)-(k-N) \psi_0 (k-N+1)+2 k \psi_0 (k+1)\bigg) \label{eq:SFF(k,b)},
   \end{align}
   where $b \in \mathbb{N}$, $k \in \mathbb{R}$, and $\Theta(x)$ denotes the unit step function defined as
   \begin{align*}
       \Theta(x)&=\begin{cases}
           1, &  \quad x \geq 0 \\
           0, & \quad x < 0.
       \end{cases}
   \end{align*}
\end{proposition}
The curve shown in Figure~\ref{fig:Third-order SFF} represents the third-order SFF, obtained as the Fourier transform of the triple-density correlation function at the point $(k, \lfloor{k}\rfloor)$. Unlike the usual second-order SFF, which captures pairwise correlations of eigenvalues, the third-order version probes higher-order statistical structure. The curve illustrates how these correlations evolve as a function of the Fourier variable $k$, starting at zero and gradually approaching saturation. Its oscillatory nature reflects the collective behavior among eigenvalues beyond simple pairwise interactions, highlighting the detailed structure of spectral correlations encoded in higher-order statistics.
\begin{figure}[H]
    \centering
    \includegraphics[width=0.6\linewidth]{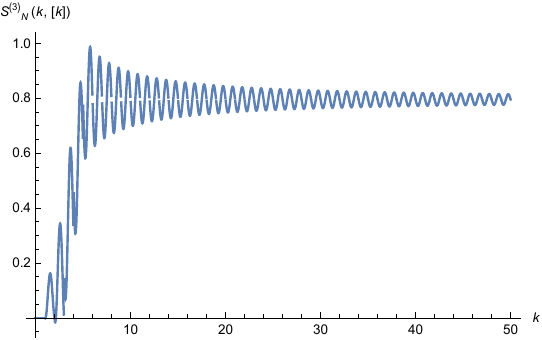}
    \caption{The third order SFF as given by Eq.~(\ref{eq:SFF(k,b)}) evaluated at the point $(k, \lfloor{k}\rfloor)$ with $N=5$.} 
    \label{fig:Third-order SFF}
\end{figure} 
\section{Proofs of Main Results} 
\label{sec:Proofs of the main results}
In this section, we prove the main results presented in Section~\ref{sec:Main results}. The proof proceeds by first expressing the CUE kernel $K_N(\theta, \phi)$ in its summation form as given in Eq.~(\ref{CUE_kernel}). Then performing the integrations in Eq.~(\ref{eq:Second order SFF CUE}) and Eq.~(\ref{eq:Third-Order_SFF}) produces expressions for the second-order and third-order SFF in terms of double and triple sums respectively. Then Eq.~(\ref{eq:Telescopic polygamma}) is used to introduce the digamma function into the double- and triple-sum expressions, which are further simplified using the finite summation identities provided in Eq.~(\ref{eq:sum_identity_psi^2}) and Eq.~(\ref{eq:sum_identity_psi}) to obtain the desired results.

The following lemma is useful in the proof of Proposition~\ref{thm:SFF}.
\begin{lemma} \label{lemma:T(k,N)}
    For $k \in \mathbb{R}\setminus  \{ -(N-1),...,(N-1) \}$ and $N \in \mathbb{N}$, let
    \begin{equation}
        T(k,N):=\sum_{\alpha=-(N-1)}^{N-1} \frac{N-|\alpha|}{k-\alpha}.
    \end{equation}
    Then we have 
    \begin{align}
        T(k,N)=(N+k)\left(\psi_0(N+k)-\psi_0(k)\right)- (N-k)\left(\psi_0(N-k)-\psi_0(-k)\right)-\frac{N}{k},
    \end{align}
    where $\psi_0(x)$ denotes the digamma function.
\end{lemma}
\begin{proof}
    The derivation goes as follows :
\begin{align}
   \sum_{\alpha=-(N-1)}^{N-1} \frac{N-|\alpha|}{k-\alpha}&=\sum_{\alpha=0}^{N-1} \frac{N-\alpha}{k+\alpha}+\sum_{\alpha=0}^{N-1} \frac{N-\alpha}{k-\alpha} -\frac{N}{k} \nn\\
    &= (N+k)\sum_{\alpha=0}^{N-1} \frac{1}{(k+\alpha)}-(N-k)\sum_{\alpha=0}^{N-1} \frac{1}{(-k+\alpha)}-\frac{N}{k} \label{Eqn_T_1}\\
    &= (N+k)\left(\psi_0(N+k)-\psi_0(k)\right)- (N-k)\left(\psi_0(N-k)-\psi_0(-k)\right)\nn\\
    &\quad-\frac{N}{k}, \label{Eqn_T_2}
\end{align}
where in Eq.~(\ref{Eqn_T_1}), we have inserted Eq.~(\ref{eq:Telescopic polygamma}) to arrive at Eq.~(\ref{Eqn_T_2}). This completes the proof.
\end{proof}

\subsection{Proof of Proposition~\ref{thm:SFF}}
Here, we provide a proof of Proposition~\ref{thm:SFF}.
\begin{proof}
Since the second order SFF is symmetric in $k$, it is sufficient to assume $k \geq 0$ for the proof. In terms of the CUE kernel, the second-order SFF, as defined in Eq.~(\ref{eq:Second order SFF CUE}), can be expressed as
\begin{align}
    S_N(k) &= N- \int_{0}^{2\pi} \int_{0}^{2\pi}  e^{i k (\theta -\phi)} K_{N}(\theta,\phi) K_{N}(\phi, \theta) \dd \theta \dd \phi,
\end{align}
which upon inserting the summation form of the kernel $K_N(\theta,\phi)$ as given in Eq.~(\ref{CUE_kernel}), and performing the integration becomes
\begin{align}
     S_N(k)
    &=N-\frac{1}{4\pi^2}\sum_{l,m=0}^{N-1} \int_{0}^{2\pi} \int_{0}^{2\pi} e^{i(k+l-m)(\theta-\phi)} \dd \theta \dd \phi\label{A_1}\\
    &=
    \begin{cases}
        N- \frac{1}{2\pi^2} \left(1-\cos(2\pi k)\right) \sum_{l,m=0}^{N-1} \frac{1}{(k+l-m)^{2}}, & \quad k \notin \{-(N-1),...,(N-1) \} \\
         \min\left(|k|,N\right), & \quad k \in \{-(N-1),...,(N-1) \}. 
    \end{cases}
     \label{A_2}
\end{align}

First, we consider the case $k \notin \{-(N-1),...,N-1 \}$. To rewrite the summation in Eq.~(\ref{A_2}) in terms of a new variable $\alpha:=l-m$, we need to count the number of solutions of the equation $l-m=\alpha$ with the constraints $0 \leq l,m \leq N-1$ and $-(N-1) \leq \alpha \leq N-1$. Let us assume that $\alpha \geq 0$. Possible solutions for the pair $(l,m)$ are $\{(\alpha,0), (\alpha+1,1),...,(\alpha+m_0,m_0)\}$, where $\alpha+m_0=N-1$. Hence, the number of solutions is $m_0+1=N-\alpha$. When $\alpha <0$, one can find similarly that the number of solutions is $N-|\alpha|$. Now, writing the summation in Eq.~(\ref{A_2}) in terms of $\alpha$,  we can express it as
\begin{align}
    S_N(k)&= N- \frac{1}{4\pi^2} 2\left(1-\cos(2\pi k)\right) \sum_{\alpha=-(N-1)}^{N-1} \frac{N-|\alpha|}{(k+\alpha)^{2}}\nn\\ 
    &=\frac{\sin^2(\pi |k|)}{\pi^2} \Big( \psi_0(N+|k|) + (N+|k|) \psi_{1}(N+|k|) + \psi_0(N-|k|) \nn\\
    &\quad+ (N-|k|) \psi_{1}(N-|k|) + |k| \psi_{1}(-|k|)-|k| \psi_{1}(|k|)- \psi_0(|k|)-\psi_0(-|k|)\Big) \label{C_1},
\end{align}
where we have used Lemma~(\ref{lemma:T(k,N)}) and the fact that $S_N(k)$ is symmetric in $k$, to arrive at the last equation. Using the reflection formula~(\ref{eq:reflection}), we rewrite Eq.~(\ref{C_1}) as
 \begin{equation}
    S_N(k)= 
       \begin{cases}
             |k|- \frac{\sin(2\pi |k|)}{2\pi}+\frac{\sin^2(\pi |k|)}{\pi^2} \Big( \psi_0(N+|k|) + (N+|k|) \psi_{1}(N+|k|) \\+ \psi_0(N-|k|) + (N-|k|) \psi_{1}(N-|k|)
    -2|k| \psi_{1}(|k|)- 2\psi_0(|k|)\Big),   &  \quad  |k|< N   \\ 
             N+\frac{\sin^2(\pi |k|)}{\pi^2} \Big( \psi_0(N+|k|) + (N+|k|) \psi_{1}(N+|k|) \\+\psi_0(|k|-N)+(|k|-N)\psi_{1}(|k|-N)
    -2|k| \psi_{1}(|k|)- 2\psi_0(|k|)\Big),   & \quad |k|>N,
       \end{cases} \label{SS_1}
 \end{equation}
which can be compactly written as follows
 \begin{align}
 \nn
    S_N(k)= \mathrm{min}\bigg(N,|k|- \frac{\sin(2\pi |k|)}{2\pi}\bigg)+\frac{\sin^2(\pi |k|)}{\pi^2} \Big( \psi_0(N+|k|) + (N+|k|) \psi_{1}(N+|k|) \\+ \psi_0\Big(\big|N-|k|\big|\Big) + \Big(\big|N-|k|\big|\Big) \psi_{1}\Big(\big|N-|k|\big|\Big)
    -2|k| \psi_{1}(|k|)- 2\psi_0(|k|)\Big). \label{AA_1}
 \end{align}
 
 When $k \in \{ -(N-1),...,(N-1) \}$, all the term in the summation of Eq.~(\ref{A_1}) vanishes after performing the integration except the one with $\alpha=k$, which gives
\begin{align}
    S_{N}(k)&=N-\frac{1}{4 \pi^2}4 \pi^2(N-|k|)=|k|. \label{SS_2}
\end{align}
Although, in Eq.~(\ref{AA_1}), the terms $\psi_0(|k|)$ and $ \psi_{1}(|k|)$ have singularity at $|k|=0$, and the terms $ \psi \left(\big|N-|k|\big|\right)$ and $ \psi_{1}\left(\big|N-|k|\big|\right)$ have singularity at $|k|=N$, the limits
\begin{align}
    &\lim_{k\to 0} \psi_0(|k|)+|k|\psi_1(|k|),\\
    &\lim_{k\to N} \psi_0(|N-|k||)+|N-|k||\psi_1(|N-|k||)
\end{align}
exist, and can be found using the relation
 \begin{align}
     x \psi_{1}(x)+\psi_0(x)=x \psi_{1}(x+1)+\psi_0(x+1). \label{SHW_4}
 \end{align}
Applying the relation~(\ref{SHW_4}), we combine  Eq.~(\ref{AA_1}) and ~(\ref{SS_2}) to arrive at the following final form of the second-order SFF
 \begin{align}
 \nn
    S_N(k)= \mathrm{min}\bigg(N,|k|- \frac{\sin(2\pi |k|)}{2\pi}\bigg)+\frac{\sin^2(\pi |k|)}{\pi^2} \Big( \psi_0(N+|k|) + (N+|k|) \psi_{1}(N+|k|) \\+ \psi_0\Big(\big|N-|k|\big|+1\Big) + \Big(\big|N-|k|\big|\Big) \psi_{1}\Big(\big|N-|k|\big|+1\Big)
    -2|k| \psi_{1}(|k|+1)- 2\psi_0(|k|+1)\Big).
 \end{align}
 This completes the proof.
\end{proof}
\subsection{Proof of Corollary~\ref{prop:k=Nt asymptotic}}
The proof of Corollary~\ref{prop:k=Nt asymptotic} is provided below.
\begin{proof}
     As $x \to \infty$, the asymptotic expansions of $\psi_0(x)$ and $\psi_{1}(x)$ up to $\mathcal{O}(x^{-2})$ and $\mathcal{O}\left(x^{-3}\right)$ are respectively given by~\cite{Luke}  $\psi_0(x)=\ln x-\frac{1}{2x}+\mathcal{O}(x^{-2})$ and $\psi_{1}(x)=\frac{1}{x}+\frac{1}{2x^2}+\mathcal{O}\left(x^{-3}\right)$, which yields  
\begin{align}
    \psi_0(x)+x \psi_{1}(x) =1+\ln x+\mathcal{O}\left(x^{-2}\right) \label{X1},
\end{align}  
and 
\begin{align}
    \psi_0(x+1)+x \psi_{1}(x+1)
    &=1+\ln(x+1)-\frac{1}{x+1}+\mathcal{O}\left(x^{-2}\right)\label{X2}.
\end{align}  
    Since $k$ is of the order of $N$, i.e., $k=t N$ with $t$ being the local coordinate, using Eq.~(\ref{X1}), we have
    \begin{align}
    \psi_0(N+|k|) + (N+|k|) \psi_{1}(N+|k|)&=\psi_0(N(1+|t|)) + (N(1+|t|)) \psi_{1}(N(1+|t|)) \nn\\
    &= 1+ \ln(N(1+|t|))+\mathcal{O}\left(N^{-2}\right) \label{X3}.
\end{align}
When $t\neq 1$, Eq~(\ref{X2}) leads to
\begin{align}
    &\psi_0\Big(1+\big|N-|k|\big|\Big) + \big|N-|k|\big| \psi_{1}\Big(1+\big|N-|k|\big|\Big) \nn\\
    &=\psi_0\Big(1+N\big|1-|t|\big|\Big) + N\big|1-|t|\big| \psi_{1}\Big(1+N\big|1-|t|\big|\Big) \nn\\
    &= 1+\ln(N\big|1-|t|\big|+1)-\frac{1}{N\big|1-|t|\big|+1}+\mathcal{O}\left(N^{-2}\right), \label{X4}
\end{align}
and for $t \neq 0$, Eq~(\ref{X2}) leads to
\begin{align}
    \psi_0(|k|+1) + |k|\psi_{1}(|k|+1)&=\psi_0(N|t|+1) + N|t|\psi_{1}(N|t|+1) \nn\\
    &= 1+ \ln(N|t|+1)-\frac{1}{N|t|+1}+\mathcal{O}\left(N^{-2}\right)\label{X5}.
\end{align}
Inserting Eq.~(\ref{X3}), Eq.~(\ref{X4}), and Eq.~(\ref{X5}) in Eq.~(\ref{exact_SFF_1}), and simplifying, we obtain
\begin{align*}
  S_N(tN)
  &=  \min\left(N,N|t|- \frac{\sin(2\pi N|t|)}{2\pi}\right)+\frac{\sin^2(\pi N|t|)}{\pi^2} \Bigg( \ln \frac{\big|1-|t|^2\big|}{|t|^2} \nn\\
  &\quad+\frac{1}{N\big|1-|t|\big|}- \frac{2}{N|t|}-\frac{1}{N\big|1-|t|\big|+1} +\frac{2}{N|t|+1}+\mathcal{O}\left(N^{-2}\right)\Bigg).
\end{align*}
As $\frac{1}{x}-\frac{1}{1+x}=\frac{1}{x^2}+\mathcal{O}\left(x^{-3}\right)$, we have
 \begin{align}
      S_N(tN)
  &=  \min\left(N,N|t|- \frac{\sin(2\pi N|t|)}{2\pi}\right)+\frac{\sin^2(\pi N|t|)}{\pi^2} \Bigg( \ln \frac{\big|1-|t|^2\big|}{|t|^2}+\mathcal{O}\left(N^{-2}\right)\Bigg).
 \end{align}
 The above equation is valid for $t \in \mathbb{R} \setminus \{0,1 \}$. For $t \in \{0,1 \}$, Eq.~(\ref{exact_SFF_1}) immediately yields,
 \begin{align}
     S_N(tN)=\mathrm{min}(N,N |t|).
 \end{align}
 Now, with the observation
 \begin{align}
     \lim_{t \to t_0} \sin^2(\pi N|t|)  \ln \frac{\big|1-|t|^2\big|}{|t|^2}=0, \quad t_0 \in \{0,1\},
 \end{align}
 we conclude the proof.
 \end{proof}
\subsection{Proof of Corollary~\ref{prop:k_O(1)}}
The proof of Corollary~\ref{prop:k_O(1)} is provided below.
\begin{proof}
    From Eq.~(\ref{X1}) and~(\ref{X2}), we have 
    \begin{align}
         \psi_0(N+|k|) + (N+|k|) \psi_{1}(N+|k|)&=1+ \ln(N+|k|)+\mathcal{O}\left(N^{-2}\right),
    \end{align}
 and 
\begin{align}
    &\psi_0\Big(1+\big|N-|k|\big|\Big) + \big|N-|k|\big| \psi_{1}\Big(1+\big|N-|k|\big|\Big)\nn\\
    &= 1+\ln(\big|N-|k|\big|+1)-\frac{1}{\big|N-|k|\big|+1}+\mathcal{O}\left(N^{-2}\right).
\end{align}   
Inserting the above equations in Eq.~(\ref{exact_SFF_1}) leads to
\begin{align}
     S_N(k)&= \mathrm{min}\bigg(|k|- \frac{\sin(2\pi |k|)}{2\pi},N\bigg)+\frac{\sin^2(\pi |k|)}{\pi^2} \bigg( 1+ \ln(N+|k|) \nn \\
     &\quad+ 1+\ln \left(\big|N-|k|\big|+1\right)-\frac{1}{\big|N-|k|\big|+1}
    -2|k| \psi_{1}(|k|+1)- 2\psi_0(|k|+1)\nn\\
    &\quad+\mathcal{O}\left(N^{-2}\right)\bigg)\nn\\
     &=\frac{\sin^2(\pi |k|)}{\pi^2} \Big( 2+ 2\ln N
    -2|k| \psi_{1}(|k|+1)- 2\psi_0(|k|+1)\Big)\nn\\
    &\quad+\mathrm{min}\bigg(|k|- \frac{\sin(2\pi |k|)}{2\pi},N\bigg)+\mathcal{O}\left(N^{-2}\right) \label{SYW}.
\end{align}
Since $k=\mathcal{O}(1)$ and $N \to \infty$, we have
\begin{align}
    \mathrm{min}\bigg(|k|- \frac{\sin(2\pi |k|)}{2\pi},N\bigg)= |k|- \frac{\sin(2\pi |k|)}{2\pi} \label{SYW1}
\end{align}
Inserting Eq.~(\ref{SYW1}) in Eq.~(\ref{SYW}) completes the proof.
\end{proof}
\subsection{Proof of Corollary~\ref{prop:Arbitrary_point_x_0}}
The proof of Corollary~\ref{prop:Arbitrary_point_x_0} is provided below.
\begin{proof}
The asymptotic expansion of $\psi_0(x)$ and $\psi_{1}(x)$ up to $\mathcal{O}(x^{-4})$ and $\mathcal{O}(x^{-5})$, respectively, are given by~\cite{Luke}
\begin{align}
    \psi_0(x)&=\ln x-\frac{1}{2x}-\frac{1}{12x^2}+\mathcal{O}\left(x^{-4}\right),\\
    \psi_{1}(x)&=\frac{1}{x}+\frac{1}{2x^2}+\frac{1}{6x^3}+\mathcal{O}\left(x^{-5}\right),
\end{align}
which implies that 
\begin{align}
    \psi_0(x)+x \psi_{1}(x) =1+\ln(x)+\frac{1}{12 x^2}+\mathcal{O}\left(x^{-4}\right). \label{eq:asymptotic_4th_power}
\end{align} 
In the following, we analyze three different cases separately.
\\

Case-I: Here we consider the case $0 <x_0<1$. 
From Eq.~(\ref{eq:asymptotic_4th_power}), we have 
\begin{align}
    &\psi_0(N+\modu{Nx_0+\tau})+\left(N+\modu{Nx_0+\tau}\right) \psi_{1}(N+\modu{Nx_0+\tau})\nn\\
    &=1+\ln N+\ln(1+x_0)+\frac{\tau}{(1+x_0)N}+\frac{(1-6\tau^2)}{12(1+x_0)^2N^2}\nn\\
    &\quad+\frac{(2\tau^3-\tau)}{6(1+x_0)^3N^3}+\mathcal{O}\left(N^{-4}\right), \label{eq:Asymp_1}
\end{align}

\begin{align}
    &\psi_0(\modu{N-\modu{N x_0+\tau}})+\left(\modu{N-\modu{N x_0+\tau}}\right) \psi_{1}(\modu{N-\modu{N x_0+\tau}}) \nn\\
    &=1+\ln N+\ln(1-x_0)-\frac{\tau}{(1-x_0)N}+\frac{(1-6\tau^2)}{12(1-x_0)^2N^2}\nn\\
    &\quad-\frac{(4\tau^3-2\tau)}{12(1-x_0)^3N^3}+\mathcal{O}\left(N^{-4}\right),  \label{eq:Asymp_2}
\end{align}
and 
\begin{align}
    &\psi_0(\modu{N x_0+\tau})+\left( \modu{N x_0+\tau}\right) \psi_{1}(\modu{N x_0+\tau})\nn\\ 
    &=1+\ln N+\ln x_0+\frac{\tau}{x_0 N}+\frac{1-6\tau^2}{12x_0^2N^2}+\frac{2\tau^3-\tau}{6x_0^3N^3}+\mathcal{O}\left(N^{-4}\right).   \label{eq:Asymp_3}
\end{align}
Inserting Eqs.~(\ref{eq:Asymp_1}), (\ref{eq:Asymp_2}) and (\ref{eq:Asymp_3}) in Eq.~(\ref{AA_1}) gives 
\begin{align}
    S_N\left(N x_0+\tau \right)
    &= x_0+\frac{\tau}{N}- \frac{\sin(2\pi (N x_0+\tau))}{2\pi N}+\frac{\sin^2(\pi (N x_0+\tau))}{\pi^2 N} \Bigg( \ln \frac{1-x_0^2}{x_0^2}\nn\\
    &\quad+ \frac{t}{N}\B{\frac{1}{1+x_0}-\frac{1}{1-x_0}-\frac{2}{x_0}}+\frac{1-6\tau^2}{12N^2}\bigg(\frac{1}{(1+x_0)^2}+\frac{1}{(1-x_0)^2}\nn\\
    &\quad-\frac{2}{x_0^2}\bigg)+\frac{4\tau^3-2\tau}{12N^3}\B{\frac{1}{(1+x_0)^3}-\frac{1}{(1-x_0)^3}-\frac{2}{x_0^3}}+\mathcal{O}\left(N^{-4}\right)\Bigg). 
\end{align}
Using the above expansion of $S_N\left(N x_0+\tau \right)$ in Eq.~(\ref{eq:Fluctuation}), we obtain 
\begin{align*}
    \Delta(x_0,\tau)
    &=\tau- \frac{\sin(2\pi (N x_0+\tau))}{2\pi}+\frac{\sin^2(\pi (Nx_0+\tau))}{\pi^2 } \Bigg( \ln \frac{1-x_0^2}{x_0^2}\nn\\
    &\quad + \frac{\tau}{N}\B{\frac{1}{1+x_0}-\frac{1}{1-x_0}-\frac{2}{x_0}}+\frac{1-6\tau^2}{12N^2}\B{\frac{1}{(1+x_0)^2}+\frac{1}{(1-x_0)^2}-\frac{2}{x_0^2}}\nn\\
    &\quad+\frac{2\tau^3-\tau}{6N^3}\B{\frac{1}{(1+x_0)^3}-\frac{1}{(1-x_0)^3}-\frac{2}{x_0^3}}+\mathcal{O}\left(N^{-4}\right)\Bigg). 
\end{align*}
\\

Case-II: 
Here, we consider the case $x_0=0$, for which the asymptotic expansions in Eq.~(\ref{eq:Asymp_1}) and (\ref{eq:Asymp_2}) remain valid, however the expression in Eq.~(\ref{eq:Asymp_3}) is no longer applicable. The appropriate expression is given by
\begin{align}
   \psi_0(\modu{N x_0+\tau})+\left(\modu{N x_0+\tau} \right) \psi_{1}(\modu{N x_0+\tau})&= \psi_0(\modu{\tau})+ \modu{\tau} \psi_{1}(\modu{\tau}). \label{SHW_2}
\end{align}
Inserting Eqs.~(\ref{eq:Asymp_1}), (\ref{eq:Asymp_2}) with $x_0=0$, along with Eq.~(\ref{SHW_2}) into Eq.~(\ref{AA_1}), we have
\begin{align*}
     \Delta(0,\tau)
    &=\frac{\sin^2(\pi |\tau|)}{\pi^2 } \bigg(2+2\ln N+\frac{1-6\tau^2}{6N^2}-2\psi_0(\modu{\tau})-2 \modu{\tau} \psi_{1}(\modu{\tau})+\mathcal{O}\left(N^{-4}\right)\bigg) \\
    &\quad+|\tau|- \frac{\sin(2\pi \tau)}{2\pi}.
\end{align*}
\\

Case-III:
When $x_0=1$, the asymptotic expansions given in Eq.~(\ref{eq:Asymp_1}) and (\ref{eq:Asymp_3}) are valid, however the expansion in Eq.~(\ref{eq:Asymp_2}) is no longer applicable. The correct expression is given by
\begin{align}
    &\psi_0(\modu{N-\modu{N x_0+\tau}})+\left(\modu{N-\modu{N x_0+\tau}} \right) \psi_{1}(\modu{N-\modu{N x_0+\tau}})\nn\\
    &=\psi_0(\modu{\tau})+\modu{\tau} \psi_{1}(\modu{\tau}). \label{SHW_3}
\end{align}
Inserting Eqs.~(\ref{eq:Asymp_1}), (\ref{eq:Asymp_3}) with $x_0=1$, and Eq.~(\ref{SHW_3}) into Eq.~(\ref{AA_1}), we have
\begin{align*}
    \Delta(1,\tau)
     &=\tau- \frac{\sin(2\pi \tau)}{2\pi}+\frac{\sin^2(\pi \tau)}{\pi^2 } \bigg(-1-\ln N+ \ln 2- \frac{3\tau}{2N}-\frac{7(1-6\tau^2)}{48N^2}\\
    &\quad-\frac{5(2\tau^3-\tau)}{16N^3}+\psi_0(\modu{\tau})+\modu{\tau} \psi_{1}(\modu{\tau})\bigg)+\mathcal{O}\left(N^{-4}\right)
\end{align*}
when $\tau <0$,
and
\begin{align*}
    \Delta(1,\tau) &=\frac{\sin^2(\pi \tau)}{\pi^2 } \bigg(-1-\ln N+ \ln 2- \frac{3\tau}{2N}-\frac{7(1-6\tau^2)}{48N^2}\\
    &\quad-\frac{5(2\tau^3-\tau)}{16N^3}+\psi_0(\modu{\tau})+\modu{\tau} \psi_{1}(\modu{\tau})\bigg)+\mathcal{O}\left(N^{-4}\right)
\end{align*}
when $\tau >0$.
This completes the proof.
\end{proof}
\subsection{Proof of Proposition~\ref{prop: 3rd_order_SFF_at_(k,-k)}}
Here, we present the proof of Proposition~\ref{prop: 3rd_order_SFF_at_(k,-k)}.
\begin{proof}
The computation of the integral in Eq.~(\ref{eq:Third-Order_SFF}) gives
\begin{align}
\mathrm{Re}\left(S_{N}^{(3)}(k,-k)\right)&=\frac{2 - 2 \cos(2 \pi k)}{(2 \pi)^3} 
\Bigg(
\sum_{l=0}^{N-1}\sum_{m=0}^{N-1}\sum_{p=0}^{N-1} 
\frac{1}{(k+l-p)(k+l-m)}
\nn\\
&\quad+
\sum_{l=0}^{N-1}\sum_{m=0}^{N-1}\sum_{p=0}^{N-1} 
\frac{1}{(k-l+p)(k-l+m)}
\Bigg)+ \frac{N^2(1-\cos(2 \pi k))}{k \pi}\,\psi_{0}(k) \nonumber \\
&\quad - \frac{(k-N)\sin(k \pi)}{k \pi^3}
\left( 4 k \pi^3 \cos(k \pi) + N(1-2 \pi^2)\sin(k \pi)\right)
B(-N,k) \nonumber \\
&\quad + (k+N)\sin(k \pi)\left(-4\cos(k \pi) 
+ \frac{N(-1+2 \pi^2)}{k \pi^3}\sin(k \pi)\right)
B(N,k) \nonumber \\
&\quad + \frac{N(1-\cos(2 \pi k))}{\pi}B(N,-k)- \frac{N^2(1-\cos(2 \pi k))}{k \pi}\,\psi_{0}(k-N) \nonumber \\
&\quad- \frac{N(1-\cos(2 \pi k))}{k \pi}\left((2k+N)A(-k)-(k+N)B(-N,-k) \right)
 \nonumber \\
&\quad + 2 \sin(k \pi) \left( 4k \cos(k \pi) 
+ \frac{N(1-2 \pi^2)}{\pi^3}\sin(k \pi)\right) 
A(k) \nonumber \\
&\quad + \frac{N}{2\pi} - 2N\pi 
- \frac{N^3(2-2\cos(2 \pi k))}{8k^2 \pi^3},
\end{align}
where $A(k):=\psi_{0}(1+k)-\psi_{0}(1)$ and $B(N,k):=\psi_{0}(1+k+N)-\psi_{0}(1)$. Using the fact that summing in the reverse order corresponds to the transformations $l \mapsto (N-1-l)$,  $m \mapsto N-1-m$ and  $p \mapsto N-1-p$, we have
\begin{align}
    \sum_{l=0}^{N-1} \sum_{m=0}^{N-1} \sum_{p=0}^{N-1} 
\frac{1}{(k+l-m)(k+l-p)}&= \sum_{l=0}^{N-1} \sum_{m=0}^{N-1} \sum_{p=0}^{N-1} 
\frac{1}{(k-l+m)(k-l+p)} \nonumber \\
&=\sum_{l=0}^{N-1} \left(\psi_{0}(k+1+l)-\psi_{0}(k-N+1+l)\right)^2, \label{eq:CUE2}
\end{align}
where, to arrive at Eq.~(\ref{eq:CUE2}), we have used the difference equation for digamma function as given in Eq.~(\ref{eq:Telescopic polygamma}).
Using the finite summation identity~\cite{Milgram,Lu_Wei_2,Lu_Wei_1,Lu_Wei_3}
\begin{align}
\sum_{i=1}^{N} \psi_{0}^2(a+i)
=&(a+N) \psi _0^2(a+N){}-(2 a+2 N-1) \psi _0(a+N)-a \psi _0^2(a){}\nonumber\\
&+(2 a-1) \psi _0(a)+2 N, \label{eq:sum_identity_psi^2}
\end{align}
in Eq.~(\ref{eq:CUE2}), we obtain the desired result and complete the proof.
\end{proof}
\subsection{Proof of Proposition~\ref{3rd SFF k,floor k}}
The proof of Proposition~\ref{3rd SFF k,floor k} proceeds as follows.
\begin{proof}
A direct computation of the integral in Eq.~(\ref{eq:Third-Order_SFF}) with the fact that $b \in \mathbb{N}$ gives
\begin{center}
\scalebox{1}{$
\begin{aligned}
&\mathrm{Re}\left(S_{N}^{(3)}(k,b)\right)\\
&=-\frac{ N \sin (2 \pi  k)\, }{(2 \pi )^2}\delta_{b,0} 
\left(\frac{N^2}{k}-\sum_{p=0}^{N-1}\sum_{q=0}^{N-1}\frac{1}{k+p-q}\right)-\frac{N \sin (2 \pi  k)\,}{(2 \pi )^2}\delta_{b,0} 
\sum_{p=0}^{N-1}\sum_{q=0}^{N-1}\frac{1}{k+p-q} \\
&\quad + \frac{\sin (2 \pi  k) 
}{(2 \pi )^2}\sum_{p=\max(0,b)}^{\min(N-1,b+N-1)} \sum_{q=0}^{N-1} \frac{1}{k+p-q}+ \frac{\sin (2 \pi  k) 
}{(2 \pi )^2}\sum_{q=\max(0,-b)}^{\min(N-1,-b+N-1)} \sum_{p=0}^{N-1} \frac{1}{k+p-q} \\
&\quad - \frac{ N \sin (2 \pi  k) \bigl(N^2 \delta_{b,0}-\max(N-|b|,0)\bigr)}{(2 \pi )^2 k} - \frac{N \sin (2 \pi  k)\, \max(N-|b|,0)}{(2 \pi )^2 k}- \frac{N^2 \sin (2 \pi  k)}{(2 \pi )^2 k} 
 \\
&\quad + \left(\frac{N^2 \delta_{b,0}}{2 \pi } - \frac{\max(N-|b|,0)}{2 \pi }\right) + \frac{ N^3 \sin (2 \pi  k)\, \delta_{b,0}}{2 \pi^2k}- \frac{ N^2 \delta_{b,0}}{2 \pi } - \frac{N^2 \sin (2 \pi  k)}{(2 \pi )^2 (b+k)} \\
&\quad + \sin (2 \pi  k) \left(
\frac{N^2}{(2 \pi )^2 (b+k)} 
- \frac{1}{(2 \pi )^2}\sum_{p=0}^{N-1}\sum_{q=0}^{N-1}\frac{1}{b+k+p-q}
\right) \nonumber\\
&\quad+ \sin (2 \pi  k) \left(
\frac{N^2}{(2 \pi )^2 k} 
- \frac{1}{(2 \pi )^2}\sum_{p=0}^{N-1}\sum_{q=0}^{N-1}\frac{1}{k+p-q}
\right)+ \frac{N}{2 \pi }, \\
\end{aligned}
$}
\end{center}
which further simplifies to 
\begin{align}
&\mathrm{Re}\left(S_{N}^{(3)}(k,b)\right)\nonumber\\
&=\frac{\sin(2 \pi k)}{(2 \pi)^2} \Bigg(
\sum_{p=\max(0,b)}^{\min(N-1,\,N-1+b)} 
\Big(\psi_{0}(1+k+p)-\psi_{0}(1+k-N+p)\Big) \nonumber\\
&\quad + \sum_{q=\max(0,-b)}^{\min(N-1,\,N-1-b)} 
\Big(-\psi_{0}(k-q)+\psi_{0}(k+N-q)\Big) + 2k\,\psi_{0}(1+k) \nonumber\\
&\quad+ 2(b+k)\,\psi_{0}(1+b+k)- (k-N)\,\psi_{0}(1+k-N) \nonumber \\
&\quad - (b+k-N)\,\psi_{0}(1+b+k-N) 
- (k+N)\,\psi_{0}(1+k+N) 
\nonumber\\
&\quad- (b+k+N)\,\psi_{0}(1+b+k+N) 
\Bigg)+ \frac{N - \max\bigl(0,\,N-|b|\bigr)}{2\pi}. \nonumber
\end{align}
Now, using the summation identity~\cite{Milgram,Lu_Wei_1,Lu_Wei_2,Lu_Wei_3}
\begin{align}
\sum_{i=1}^{N}\psi_{0}(a+i)=(N+a)\psi_{0}(N+a+1)-a\psi_{0}(a+1)-N, \label{eq:sum_identity_psi}
\end{align}
we obtain the desired result, thus concluding the proof.
\end{proof}
\section{Conclusions}\label{Discussions}
The SFF, a central quantity in random matrix theory, serves as a powerful tool for probing universality in disordered quantum systems and diagnosing quantum chaos. In this work, we derived exact closed-form expressions for the second- and third-order SFFs in CUE, valid for all real values of the time parameter, and analyzed their asymptotic behavior across different regimes. For the second-order SFF, we obtained an exact closed-form expression in terms of polygamma functions. In the large-matrix limit and for integer time values, this expression reproduces the well-established results from earlier studies. When the time parameter is of order one relative to the matrix size, the second-order SFF exhibits logarithmic growth with the ensemble dimension. For the third-order SFF, we derived a closed-form expression in a special case by exploiting the translational invariance of CUE. Future work will focus on deriving closed-form expressions for the exact many-body SFF of noninteracting fermions.
\backmatter
\bmhead{Acknowledgment} The work of Lu Wei was supported by the U.S. National Science Foundation (2306968) and the U.S. Department of Energy (DE-SC0024631).


\begin{thebibliography}{99}

\bibitem{aurich1988hadamard}
Aurich, R., Matthies, C., Sieber, M., Steiner, F.: Quantum chaos of the Hadamard–Gutzwiller model. Phys. Rev. Lett. \textbf{61}, 483–486 (1988).

\bibitem{Berry}
Berry, M.V.: Semiclassical theory of spectral rigidity. Proc. R. Soc. Lond. A \textbf{400}, 229–251 (1985).

\bibitem{berry1987quantumchaology}
Berry, M.V.: Quantum chaology. Proc. R. Soc. Lond. A \textbf{413}, 183–198 (1987).

\bibitem{BH1997}
Brézin, E., Hikami, S.: Spectral form factor in a random matrix theory. Phys. Rev. E \textbf{55}, 4067–4083 (1997).

\bibitem{CES}
Cipolloni, G., Erdős, L., Schröder, D.: On the spectral form factor for random matrices. Commun. Math. Phys. \textbf{401}, 1665–1700 (2023).

\bibitem{CG2024}
Cipolloni, G., Grometto, N.: The dissipative spectral form factor for i.i.d. matrices. J. Stat. Phys. \textbf{191}, 21 (2024).

\bibitem{Nick}
Cotler, J., Hunter-Jones, N., Liu, J., et al.: Chaos, complexity, and random matrices. J. High Energy Phys. \textbf{2017}, 48 (2017).

\bibitem{eckhardt1992semiclassical}
Eckhardt, B., Fishman, S., Müller, K., Wintgen, D.: Semiclassical matrix elements from periodic orbits. Phys. Rev. A \textbf{45}, 3531–3546 (1992).

\bibitem{Forrester2021}
Forrester, P.J.: Quantifying dip–ramp–plateau for the Laguerre unitary ensemble structure function. Commun. Math. Phys. \textbf{387}, 215–235 (2021).

\bibitem{Forrester2021stat}
Forrester, P.J.: Differential identities for the structure function of some random matrix ensembles. J. Stat. Phys. \textbf{183}, 33 (2021).

\bibitem{FKZ}
Forrester, P.J., Kieburg, M., Li, S.-H., Zhang, J.: Dip–ramp–plateau for Dyson Brownian motion from the identity on U(N). Probab. Math. Phys. \textbf{5}, 321–355 (2024).

\bibitem{gutzwiller1971periodic}
Gutzwiller, M.C.: Periodic orbits and classical quantization conditions. J. Math. Phys. \textbf{12}, 343–358 (1971).

\bibitem{Haake}
Haake, F.: Quantum Signatures of Chaos. Springer, Berlin (1991).

\bibitem{HSW}
Haake, F., Sommers, H.-J., Weber, J.: Fluctuations and ergodicity of the form factor of quantum propagators and random unitary matrices. J. Phys. A \textbf{32}, 6903 (1999).

\bibitem{heller1984scars}
Heller, E.J.: Bound-state eigenfunctions of classically chaotic Hamiltonian systems: scars of periodic orbits. Phys. Rev. Lett. \textbf{53}, 1515 (1984).

\bibitem{Heisenberg1927}
Heisenberg, W.: Über den anschaulichen Inhalt der quantentheoretischen Kinematik und Mechanik. Z. Phys. \textbf{43}, 172–198 (1927).

\bibitem{Hosur2016}
Hosur, P., Qi, X.-L., Roberts, D.A., Yoshida, B.: Chaos in quantum channels. J. High Energy Phys. \textbf{2016}, 4 (2016).

\bibitem{lichtenberg1992regular}
Lichtenberg, A.J., Lieberman, M.A.: Regular and Chaotic Dynamics, 2nd edn. Springer, New York (1992).

\bibitem{Luke}
Luke, Y.L.: The Special Functions and Their Approximations. Academic Press, New York (1969).

\bibitem{Maldacena2016}
Maldacena, J., Shenker, S.H., Stanford, D.: A bound on chaos. J. High Energy Phys. \textbf{2016}, 106 (2016).

\bibitem{mehta1991random}
Mehta, M.L.: Random Matrices. Academic Press, San Diego (1991).

\bibitem{Milgram}
Milgram, M.: On some sums of digamma and polygamma functions. arXiv:0406338 (2017).

\bibitem{OLCH}
Oliviero, S.F., Leone, L., Caravelli, F., Hamma, A.: Random matrix theory of the isospectral twirling. SciPost Phys. \textbf{10}, 076 (2021).

\bibitem{Patrick_Hayden_2007}
Hayden, P., Preskill, J.: Black holes as mirrors: quantum information in random subsystems. J. High Energy Phys. \textbf{09}, 120 (2007).

\bibitem{robertson1929}
Robertson, H.P.: The uncertainty principle. Phys. Rev. \textbf{34}, 163 (1929).

\bibitem{sakurai2017}
Sakurai, J.J., Napolitano, J.: Modern Quantum Mechanics, 2nd edn. Cambridge Univ. Press, Cambridge (2017).

\bibitem{Shenker2014}
Shenker, S.H., Stanford, D.: Black holes and the butterfly effect. J. High Energy Phys. \textbf{2014}, 67 (2014).

\bibitem{stockmann1999quantumchaos}
Stöckmann, H.-J.: Quantum Chaos: An Introduction. Cambridge Univ. Press, Cambridge (1999).

\bibitem{strogatz2018nonlinear}
Strogatz, S.H.: Nonlinear Dynamics and Chaos: With Applications to Physics, Biology, Chemistry, and Engineering, 2nd edn. CRC Press (2018).

\bibitem{Lu_Wei_1}
Wei, L.: Skewness of von Neumann entanglement entropy. J. Phys. A \textbf{53}, 075302 (2020).

\bibitem{Lu_Wei_2}
Wei, L.: Proof of Vivo–Pato–Oshanin’s conjecture on the fluctuation of von Neumann entropy. Phys. Rev. E \textbf{96}, 022106 (2017).

\bibitem{Lu_Wei_3}
Huang, Y., Wei, L.: Entropy fluctuation formulas of fermionic Gaussian states. Ann. Henri Poincaré \textbf{24}, 4283–4342 (2023).

\end{thebibliography}
\end{document}